\pdfoutput=1
\documentclass[sigconf]{acmart}

\def\BibTeX{{\rm B\kern-.05em{\sc i\kern-.025em b}\kern-.08em
    T\kern-.1667em\lower.7ex\hbox{E}\kern-.125emX}}

\usepackage{booktabs} 
\usepackage{amsmath}
\usepackage{graphicx}
\usepackage{url}
\usepackage{subfig}
\usepackage[linesnumbered,ruled,noend]{algorithm2e}
\usepackage{enumitem}

\makeatletter
\newcommand*{\rom}[1]{\expandafter\@slowromancap\romannumeral #1@}
\makeatother

\newcommand{\riki}{\textsc{Raks}\xspace}
\newcommand{\wiki}{\textsc{WikiSearch}\xspace}

\acmDOI{}

\acmISBN{XXX-X-XXXXX-XXX-X}

\acmConference{}{}{} 

\begin{document}
\title{Efficient Radial Pattern Keyword Search on Knowledge Graphs in Parallel}
  
\author{Yueji Yang}
\orcid{0000-0001-7413-1925}
\affiliation{%
	\institution{National University of Singapore}
}
\email{yueji@comp.nus.edu.sg}

\author{Anthony K. H. Tung}
\affiliation{%
	\institution{National University of Singapore}
}
\email{atung@comp.nus.edu.sg}

\renewcommand{\shortauthors}{}

\begin{abstract}
Recently, keyword search on Knowledge Graphs (KGs) becomes popular. Typical keyword search approaches aim at finding a concise subgraph from a KG, which can reflect a close relationship among all input keywords. The connection paths between keywords are selected in a way that leads to a result subgraph with a better semantic score. However, such a result may not meet users' information needs because it relies on the scoring function to decide what keywords to link closer. Therefore, such a result may miss close connections among some keywords on which users intend to focus.  
In this paper, we propose a parallel keyword search engine, called \riki. 
It allows users to specify a query as two sets of keywords, namely \textit{central keywords} and \textit{marginal keywords}. Specifically, central keywords are those keywords on which users focus more. Their relationships are desired in the results. Marginal keywords are those less focused keywords. On one hand, their connections to the central keywords are desired. On the other hand, they provide additional information that helps discover better results in terms of user intents.
To improve the efficiency, we propose novel weighting and scoring schemes that boost the parallel execution during search while retrieving semantically relevant results. 
We conduct extensive experiments to validate that \riki can work efficiently and effectively on open KGs with large size and variety. 
\end{abstract}

\maketitle

\section{Introduction}\label{sec:introduction}
\begin{figure*}
	\centering
	\subfloat[Tree Result By Typical Approaches]{\includegraphics[width=3.35in]{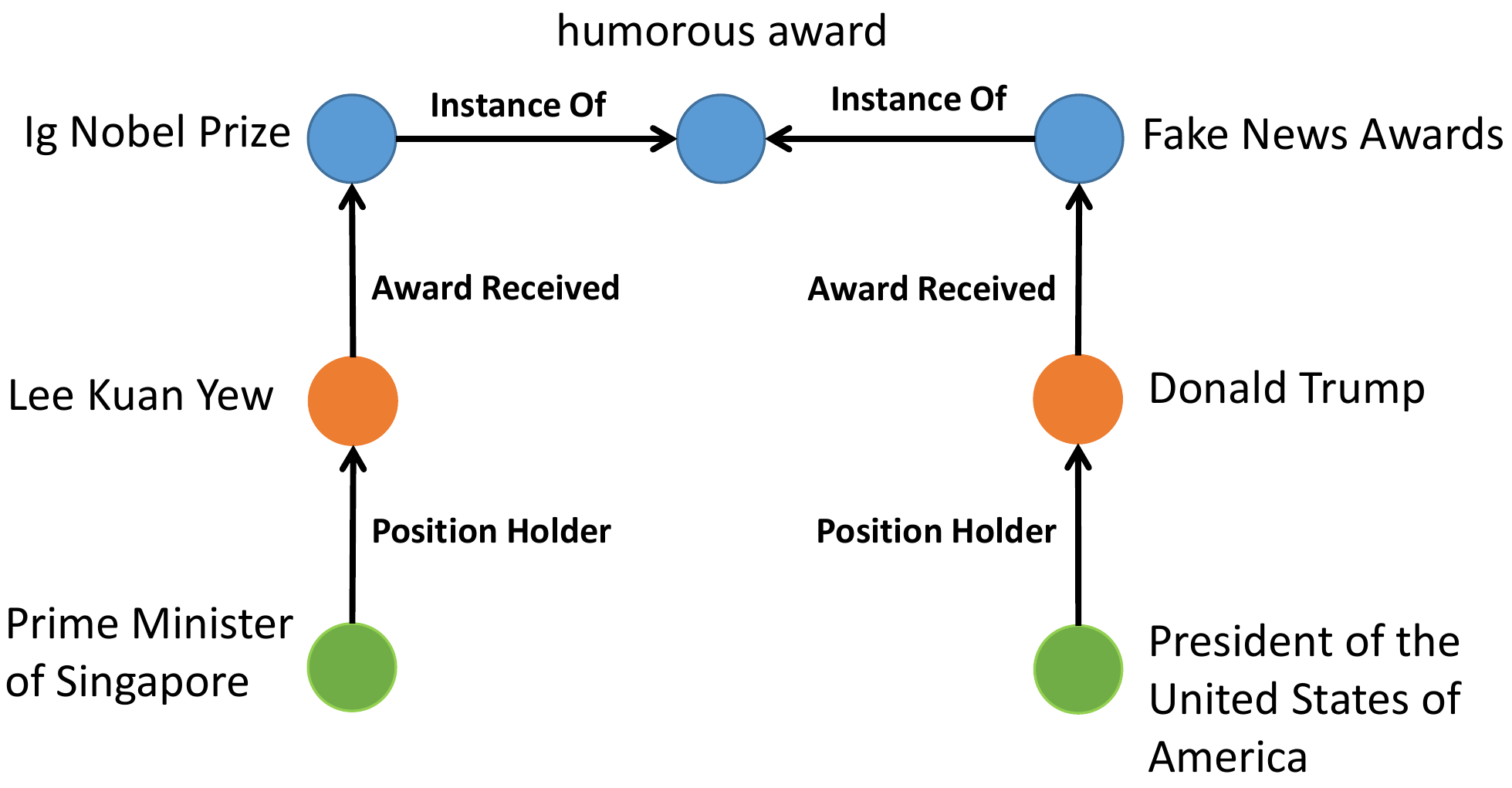}\label{subfig:introCB0} }
	\subfloat[Subgraph Result By \riki]{\includegraphics[width=3.35in]{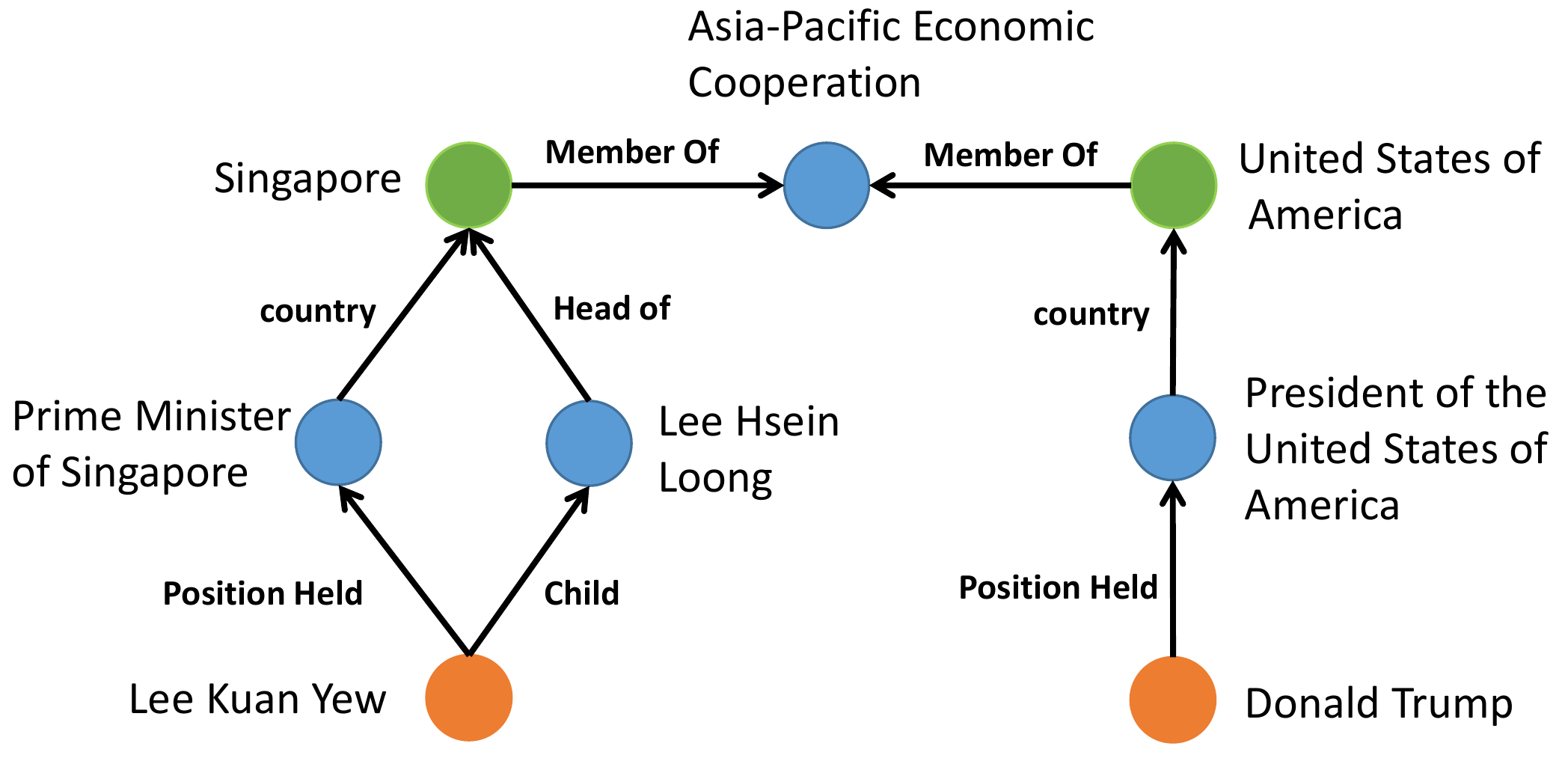}\label{subfig:introCB1}}
	
	\caption{Keyword Search results by typical keyword search approaches and \riki. The input keywords include \{\textit{Trump, Lee Kuan Yew, Singapore, USA}\}. Figure~\ref{subfig:introCB0} shows no intuitive connection between \textit{USA} and \textit{Singapore} while Figure~\ref{subfig:introCB1} does.}
	\label{fig:intro}
\end{figure*}
Open Knowledge Graph (KG) has become popular, such as Wikidata \cite{Erxleben:Wikidata} and Freebase \cite{Bollacker:freebase}. They contain various entities along with their relationships from the real world. Keyword search \cite{Wang:survey} is widely used in querying such heterogeneous data sources due to its simplicity. Users do not need to learn query languages or underlying data schema to issue a query, as they do when using structured query languages, like SQL \cite{Groff:SQL} and SPARQL \cite{sparql}. 

Typical keyword search approaches on graphs \cite{Wang:survey,Aditya:banks,He:blinks,Kacholia:bidirection} generally work as follows. Firstly, they take as input a set of keywords and find nodes containing those keywords in a KG. Then, according to some semantic ranking functions, the search engine enumerates paths from initial nodes and identifies a connected subgraph (or a tree) that relates all input keywords. Finally, top-k result subgraphs are generated and presented to users. During this process, the linking paths between keywords are enumerated and preserved based on the ranking functions, which guide the search algorithm. As a consequence, the preserved linking paths in a result may not meet users' information needs. That is, some keywords on which users focus more turn out distant in the result. Whereas, those less focused keywords which are meant to provide additional information for search may constitute the central or major portion of the result subgraph. To see this, consider the following scenario.

Suppose a user wants to know the relationships between \textit{Singapore} and \textit{USA}, and how \textit{Donald Trump} and \textit{Lee Kuan Yew} are related to or connected via the two countries.
With typical keyword search engines, the user can only submit all keywords as a whole, \{\textit{Trump, Lee Kuan Yew, Singapore, USA}\}. A top result is shown in Figure~\ref{subfig:introCB0}. From this result, there is no intuitive connection between \textit{Singapore} and \textit{USA} (in green) because they are distant in the result. Whereas, the other two less focused keywords, \{\textit{Trump, Lee Kuan Yew}\} (in orange), are connected more closely. Such a result obviously does not meet the user's information needs. This motivates us to design a way that allows users to distinguish the focused and less focused keywords for retrieving results. 
Also, the simplicity of keyword search should not be hurt.

To address the above problem, in this work, we propose to allow users to input two sets of keywords, namely \textit{central keywords} (w.r.t. focused ones) and \textit{marginal keywords} (w.r.t. less focused ones). Together, we call such a query as \textit{Radial Pattern Query} (RPQ).
Furthermore, the results w.r.t. an RPQ must capture two types of connections: one is the connections among \textit{central keywords}, the other is the connections between \textit{central} and \textit{marginal keywords}. Result subgraphs with such a connection structure tend to be in a \textit{radial pattern} shape, with \textit{central keywords} in the center and \textit{marginal keywords} surrounding the center. In particular, such result subgraphs are called \textit{Radial Pattern Graph}s (RPGs). In Figure~\ref{subfig:introCB1}, we show the corresponding result by \riki. \textit{Singapore} and \textit{USA} (in green) are submitted as \textit{central keywords} while \textit{Trump} and \textit{Lee Kuan Yew} (in orange) are regarded as \textit{marginal keywords}. From this result, users can easily learn that \textit{Singapore} and \textit{USA} are both members of \textit{Asia-Pacific Economic Cooperation}. Besides, the correspondence between \textit{Lee Kuan Yew} and \textit{Singapore} is made clear, so is that between \textit{Trump} and \textit{USA}. Note that, in Figure~\ref{fig:intro}, we use well-known people and countries for illustrative purpose. In general, we cannot assume users know the correspondences between \textit{central} and \textit{marginal keywords}, which are supposed to be discovered during search procedure.

In addition to submitting a query \textbf{manually} as in the above application scenario, \riki can also be used for \textbf{automatically} mapping documents to KGs with \textit{topic} structures preserved. 
Specifically, documents (e.g. scientific articles or books) commonly have some main topics (i.e. focused) and some relevant ones (i.e. less focused), which can be easily identified using TF-IDF scores or topic modeling like LDA \cite{Blei:LDA}. We can submit to \riki the main topics as \textit{central keywords} and the relevant ones as \textit{marginal keywords}. The result subgraphs can better capture the document topic structures. Such result subgraphs can be used to aid many downstream tasks like document clustering \cite{Wang:KnowSim}. Interestingly, \riki is rather useful in mapping news stories to a KG. A news story often has some main entities, such as \textit{Singapore} and \textit{USA} in the aforementioned example. There are also some less important but relevant entities such as \textit{Trump} and \textit{Lee Kuan Yew}. Given such a story, Figure~\ref{subfig:introCB1} serves as a good explanation for the relationships among those entities. Such relationships can not only help readers in understanding the news background, but also reveal valuable clues (e.g. \textit{Asia-Pacific Economic Cooperation} and \textit{Lee Hsein Loong} in Figure~\ref{subfig:introCB1}) for journalists to discover more newsworthy entities along with their stories \cite{Hossain:StoryTelling}.  

However, producing fast response to an RPQ is not easy, in the face of today's KGs with large volume and variety. An RPQ reduces to a typical keyword search query if there are no \textit{marginal keywords}. This indicates that solving an RPQ is at least as difficult as solving a typical keyword search query. The complexity of typical keyword search stems from Group Steiner Tree (GST) Problem \cite{Ding:min-cost-topk,Li:groupsteiner} which is NP-Hard. GST problem cannot even be approximated by a constant ratio within polynomial time \cite{Ihler}. Established works, like BANKS-\rom{1} \cite{Aditya:banks}, BANKS-\rom{2} \cite{Kacholia:bidirection} and BLINKS \cite{He:blinks}, all seek to solve the approximate version of GST problem \cite{Li:groupsteiner}. They define a scoring function which aggregates scores of paths one from each keyword. Then, they have to do many path enumeration and joining operations to discover the results. Moreover, discovering connections from \textit{marginal keywords} to \textit{central keywords} is also difficult in two aspects. First, the process inherently involves massive path enumeration and ranking as in typical keyword search. This can cause efficiency issue when the graph size is large. Second, it is not clear how to make central and marginal keywords closely connected and relevant. For example, marginal keywords may form a connected subgraph before reaching any central keywords. Such a subgraph can be irrelevant to relationships among central keywords.

We summarize our \textbf{approach} and \textbf{contributions} as follows. First, we formalize Radial Pattern Keyword Search problem, which takes as input Radial Pattern Query (RPQ) and outputs \textit{Radial Pattern Graph}s (RPGs). In addition, we also define the intermediate results as \textit{Central Graphs} (CGs), which are generated during search.
Second, we design weighting and scoring schemes, which boost the parallel execution. The goal is to avoid expensive comparison of massive path scores during search. As a result, the search can finish in time while the results are semantically relevant.
Third, we propose a unified algorithmic framework for finding both the intermediate results (CGs) and the final results (RPGs). 
Under this framework, an efficient two-phase parallel algorithm is implemented. The algorithm runs \textit{twice} with a few modifications during search for finding intermediate and final results. In the first run, it takes \textit{central keywords} as input and outputs the intermediate result subgraphs, i.e. CGs. In the second run, the input consists of \textit{marginal keywords} and the CGs obtained from first run. The final output is the top-k RPGs. 
Furthermore, this two-phase algorithm divides into \textit{exploration phase} and \textit{recovering phase}. In \textit{exploration phase}, the algorithm first explores the data graph from nodes containing keywords (\textit{central} or \textit{marginal}). Unlike typical graph traversal algorithms, this procedure is quite efficient because it need not compare and track any traversal paths. In \textit{recovering phase}, the algorithm recovers the result subgraph structures (nodes and edges) based on traversal information from the \textit{exploration phase}. More importantly, it can avoid many writing and reading conflicts in parallel. 
Furthermore, our algorithm can efficiently work on either multi-core CPUs or a single GPU. 
Lastly, we conduct extensive experiments to validate the efficiency and effectiveness of the implemented search engine, \riki. 

\vspace{1mm}
\noindent\textbf{Organization.} First, Section \ref{sec:problemDef} shows the problem definitions. Second, Section \ref{sec:weightingScoring} introduces the weighting and scoring schemes, which are essential for parallel execution. Third, Section \ref{sec:implementation} introduces details of the unified algorithmic framework and the implementation. Fourth, The experiment results are reported in Section \ref{sec:experiment}. Lastly, we discuss related work in Section \ref{sec:relatedWork} and make the conclusion in Section \ref{sec:conclusion}. 


%
\section{Problem Definition}
\label{sec:problemDef}
In this section, we introduce the preliminaries and definitions. 
We consider keyword search on a bi-directed labeled knowledge graph $\mathcal{G(V,E)}$ with entity node set $\mathcal{V}$ and relationship edge set $\mathcal{E}$. All nodes have short textual descriptions and all edges have labels. Edge labels are used to set edge weights. We added one reverse edge for every original edge in the KG since the search via reversed directions is also desirable.
A node containing input keywords is called a \textit{keyword node}.
Radial Pattern Query (RPQ) is formally presented in Definition \ref{def:RPQ}.

\begin{definition}
\textit{(Radial Pattern Query (RPQ)).} A Radial Pattern Query $\mathcal{Q=\{C,M\}}$ consists of two sets of keywords: \textit{central keywords} $\mathcal{C}=\{c_1, c_2,...,c_p\}$ and \textit{marginal keywords} $\mathcal{M}=\{m_1,m_2,...,m_q\}$, where $\mathcal{C}\neq \emptyset$.
\label{def:RPQ}
\end{definition}
When $\mathcal{M}=\emptyset$, RPQ becomes a typical keyword search query. We use $t$ to denote a keyword, without distinguishing \textit{central} or \textit{marginal}. We call a node containing central (resp. marginal) keywords as a central (resp. marginal) keyword node. 
In the following definitions, the distance from node $u$ to node $v$, denoted by $\mathcal{D}(u,v)$, is the smallest path score from $u$ to $v$. 
In this paper, the paths with smallest scores are considered shortest.
We use $\mathcal{V}_{t}$ to denote the set of all nodes containing keyword $t \in \mathcal{C}\cup\mathcal{M}$. The distance metrics between keywords and nodes are in Definition \ref{def:distKeyword2Node} and \ref{def:distKeyword2NodeSet}. By considering the smallest distances, we can find the most relevant connections.

\begin{definition}
	\textit{(Distance from a keyword to a node).} The distance from a keyword $t$ to a node $v$ is defined as
	$\mathcal{D}(t, v) = \min\limits_{u\in \mathcal{V}_t}{\mathcal{D}(u,v)}	$.
	\label{def:distKeyword2Node}
\end{definition}
\begin{definition}
	\textit{(Distance from a keyword to a node set).} The distance from a keyword $t$ to a node set $\mathcal{V}$ is defined as $\mathcal{D}(t,\mathcal{V})=\min\limits_{v\in \mathcal{V}}{\mathcal{D}(t,v)}$.
	\label{def:distKeyword2NodeSet}
\end{definition}
Before formally presenting Radial Pattern Graph (RPG), we first define the intermediate result called Central Graph (CG). It is generated by considering only \textit{central keywords}. We use $SP(u,v)$ to denote the set of all shortest path instances which produce $\mathcal{D}(u,v)$ where $u$ and $v$ are nodes. Likewise, given a keyword $t$ and a node $v$, $SP(t,v)=\bigcup SP(u^*,v)$, where $u^*\in\mathcal{V}_t$ \textit{and} $\mathcal{D}(u^*,v)=\mathcal{D}(t,v)$. Furthermore, for a keyword $t$ and a node set $\mathcal{V}$, $SP(t, \mathcal{V})=\bigcup SP(t,v^*)$, where $v^* \in \mathcal{V}$ \textit{and} $\mathcal{D}(t,v^*)=\mathcal{D}(t,\mathcal{V})$. In line with the definitions of distance metrics, these shortest paths correspond to more closely related relationships.

\begin{definition}
\textit{(Central Graph (CG)).} Given \textit{central keywords} $\mathcal{C}$ and a node $\tilde{v}\in \mathcal{V}$, a Central Graph is defined as $\mathcal{G}^c(\tilde{v},\mathcal{V}^c, \mathcal{E}^c)=\bigcup_{c_i \in \mathcal{C}} {SP(c_i, \tilde{v})}$, where $\mathcal{V}^c$ and $\mathcal{E}^c$ denote the node and edge sets in $\mathcal{G}^c$. In particular, $\tilde{v}$ is called Central Node.
\label{def:CG}
\end{definition}
In Definition \ref{def:CG}, a CG is a subgraph located at a Central Node $\tilde{v}$, connecting all \textit{central keywords}.
Different positions of $\tilde{v}$ can result in different scores even for the same-structured CGs (ignoring edge directions). We defer the discussion of scoring till Section \ref{sec:graphScoring}. 
We summarize three properties of CGs as follows.
\textbf{(1) Totality}: A CG covers all \textit{central keywords}.
\textbf{(2) Connectivity}: A CG is a connected subgraph.
\textbf{(3) Multi-path Optimality}: For every keyword, a CG keeps exactly all shortest paths from that keyword to the central node, i.e. $SP(c_i,\tilde{v})$.
CGs are meant to capture a close relationship among all \textit{central keywords}, on which users focus more. There can be many of such CGs that closely relate all \textit{central keywords}, however, only some of them are relevant when taking into consideration \textit{marginal keywords}. Although less focused, \textit{marginal keywords} are also important in two aspects. First, their connections with central keywords are desired. Second, they provide additional information to help discover results that better meet user's information needs. 

Next, we present the definition of Radial Pattern Graph (RPG), which is based on an obtained CG along with the input \textit{marginal keywords}. Let $\mathcal{V}_\mathcal{C} \subseteq \mathcal{V}^c$ denote the set of all nodes containing at least one \textit{central keyword} in a CG $\mathcal{G}^c(\tilde{v},\mathcal{V}^c,\mathcal{E}^c)$. 

\begin{definition}
\textit{(Radial Pattern Graph (RPG)).} Given a query $\mathcal{Q}=\{\mathcal{C},\mathcal{M}\}$ $(\mathcal{M}\neq\emptyset)$ and a CG $\mathcal{G}^c(\tilde{v},\mathcal{V}^c,\mathcal{E}^c)$ w.r.t. $\mathcal{C}$, a \textit{Radial Pattern Graph} is defined as $\mathcal{G}^r(\tilde{v}, \mathcal{V}^r,\mathcal{E}^r)=\mathcal{G}^c\cup \mathcal{G}^m$, where $\mathcal{G}^m = \bigcup_{m_i \in \mathcal{M}} {SP(m_i,\mathcal{V}_\mathcal{C})}$, s.t. the \textit{Pass Through Constraint} (\textbf{PTC}) is satisfied. 

\noindent\textbf{PTC}: There exists at least two different \textit{marginal keyword nodes} in an RPG such that all their simple path connections (regardless of edge directions) must pass through nodes in $\mathcal{V}_\mathcal{C}$. 
For a trivial case where $|\mathcal{M}|=1$, we only require the \textit{marginal keyword} connected to any $v\in\mathcal{V}_\mathcal{C}$.
%

\label{def:RPG}
\end{definition}

\begin{figure}
	\centering
	\includegraphics[width=3in]{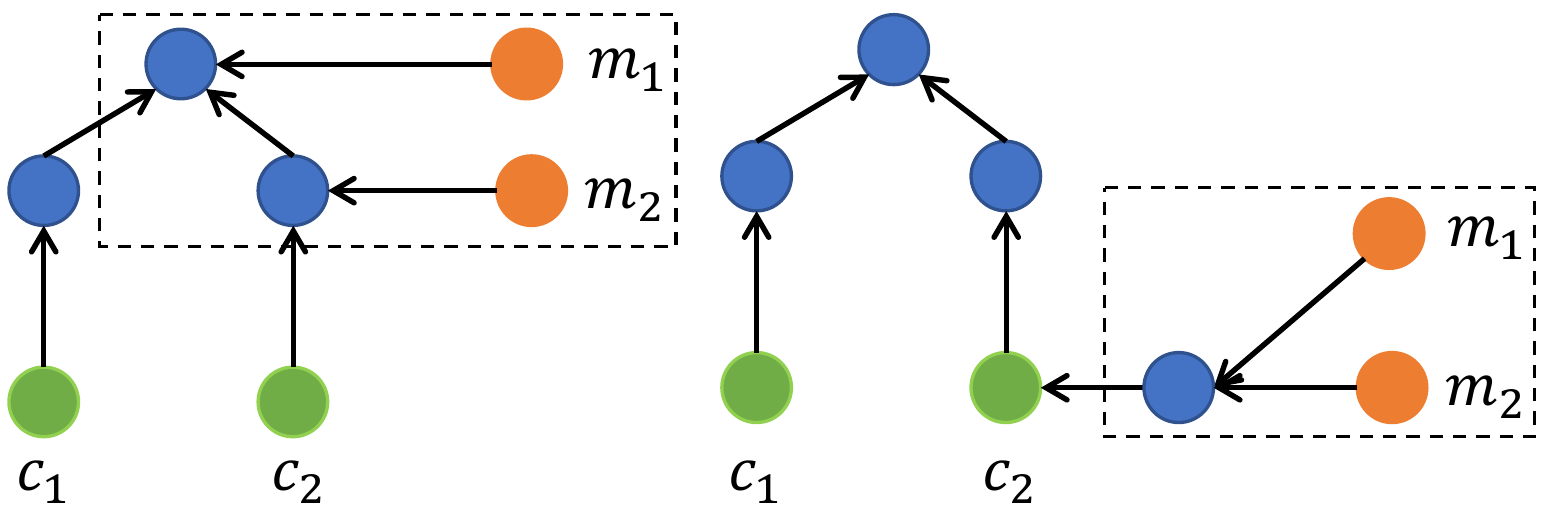}
	\caption{Examples violating PTC. The input query $\mathcal{Q}$ consists of $\mathcal{C}=\{c_1,c_2\}$ and $\mathcal{M}=\{m_1,m_2\}$.}
	\label{fig:exampleVioPTC}
\end{figure}

When context is clear, we simply use $\mathcal{G}^c$ (resp. $\mathcal{G}^r$) to denote a CG (resp. RPG).
RPG shares similar properties with CG. As for \textbf{multi-path optimality}, by including only $SP(m_i,\mathcal{V}_\mathcal{C})$ (Definition \ref{def:RPG}), we ensure a close connection from \textit{marginal keywords} to \textit{central keyword nodes} of $\mathcal{G}^c$. 
This can not only help find effective short connections, but also preserve the correspondences between \textit{central} and \textit{marginal} keywords, e.g. from \textit{Trump} to \textit{USA} and from \textit{Lee Kuan Yew} to \textit{Singapore} in Figure~\ref{subfig:introCB1}. 

In addition to path scoring, PTC helps prevent a situation where marginal keywords form a closely connected subgraph before they reach central keyword nodes. Such connections among marginal keywords are irrelevant to central keywords which users focus on. 
Figure~\ref{fig:exampleVioPTC} shows two examples violating PTC. In both examples, \textit{marginal keyword} nodes form a connected subgraph that is irrelevant to all \textit{central keyword} nodes.
Therefore, we use PTC as a structural constraint to force a close relationship between \textit{central} and \textit{marginal} keywords.
For instance, Figure~\ref{subfig:introCB1} suggests how \textit{Trump} and \textit{Lee} are connected via the relationships of the two countries, i.e. \textit{USA} and \textit{Singapore}. 
\begin{definition}
\textit{(Radial Pattern Keyword Search Problem).} Given RPQ with $\mathcal{Q}=\{\mathcal{C,M}\}$, to return top-$k$ RPGs according to some ranking scheme.
\label{def:RPKSP}
\end{definition}

By searching RPGs, users can gain insights into \textbf{(1)} relationships among \textit{central keywords}, \textbf{(2)} how they are related to \textit{marginal keywords}, and \textbf{(3)} how \textit{marginal keywords} are connected via \textit{central keywords}. RPQ allows users to specify focused and less focused keywords. Hence, the corresponding results can better meet users' information needs.
Meanwhile, the input query format remains fairly simple. In Section \ref{sec:graphScoring}, we discuss the ranking schemes mentioned in Definition \ref{def:RPKSP}.

\section{Weighting and Scoring Schemes}
\label{sec:weightingScoring}
In this section, we introduce a weighting scheme based on a graph coarsening strategy. In addition, we also design the scoring schemes of paths, CGs and RPGs. The proposed search algorithm (Section \ref{sec:implementation}) is built on these weighting and scoring schemes. In addition to improving efficiency, these schemes also help to keep the results semantically relevant and meaningful.

\subsection{Graph Coarsening}
\label{sec:edgeActivation}
Traditional keyword search methods can be seen as fine-grained. Their traversal paths are generated in a sequential manner based on carefully designed (fine-grained) edge weights. As a result, scalability issues arise due to the need to \textit{enumerate} and \textit{compare} massive path scores for calculating expansion priorities. There is little chance to parallelize the process.
However, our observation is that determining exactly the paths with better fine-grained scores may not be necessary. It wastes much time to compare paths with similar informativeness for semantic search. For example, in Figure \ref{subfig:introCB1}, the two paths from \textit{Lee Kuan Yew} to \textit{Singapore} may have slightly different fine-grained scores. Whereas, they are both relevant and informative paths. Therefore, it can save much time if we can explore massive paths with slightly different scores \textit{in parallel} without comparing them.
Inspired by this, we propose a graph coarsening and weighting scheme so that those similar paths are given the same score. 
 

Coarsening the graph includes two steps. First, we calculate the fine-grained edge weights following the traditional ways. Then, we do scaling and rounding to discretize these weights into integers (\textit{edge activation level}) using the average pairwise distance (hops) of the data graph. During the search, we use the integer weights to control search priority.

The fine-grained edge weights are computed in a way that is similar to \cite{Aditya:banks,He:blinks,Kacholia:bidirection}. The intuition is that massive edges with the same label can lead to meaningless and trivial connections. For example, two students are connected by the same department. And two people are connected because both of them are linked to the \textit{human} node by \textit{instance-of} edge. 
We use $w_{ij}$ to denote the weight of a directed edge $e_{ij}$ from node $v_i$ to $v_j$. Also, we use $l_{e_{ij}}$ to denote the label of $e_{ij}$. The weight is given by
$w_{ij}=\log (|\{e_{ix}|l_{e_{ix}}=l_{e_{ij}} \wedge v_x \in \mathcal{N}_o(v_i)\}|+|\{e_{xj}|l_{e_{xj}}=l_{e_{ij}}\wedge v_x \in \mathcal{N}_i(v_j)\}|)$, where $\mathcal{N}_o(\cdot)$ and $\mathcal{N}_i(\cdot)$ denote the sets of out-neighbors and in-neighbors. We use log to scale the number of edges \cite{Aditya:banks,Kacholia:bidirection}, since some nodes may have thousands or even millions of edges with the same label.
We then rescale the edge weights into [0,1] for ease of processing by $w_{ij} = \frac{w_{ij}-w_{min}}{w_{max}-w_{min}}$.

The next step is to coarsen the edge weights. We adapt the \textit{reward-and-penalty} strategy introduced in \cite{Yang:wikisearch} with two major modifications. First, we use edge weighting instead of node weighting by \cite{Yang:wikisearch}. Second, we develop bounds on the orginal fine-grained weights, which is missing in \cite{Yang:wikisearch}.

The \textit{reward-and-penalty} strategy scales and discretizes the edge weights into consecutive integers around the average shortest hops of the data graphs. The reason for using average pair-wise hops is to get a reasonable range for the coarsened integer weights to guide the search. This will be made clearer shortly when we present the definition of the path scoring in Section \ref{sec:pathScoring}. In the scaling process, we design a parameter $\alpha \in (0,1)$ to control the extent of scaling. Let $\overline{A}$ denote the average shortest hops in the graph. We refer to $a_{ij}$ as \textit{edge activation level}.


\begin{equation}
\label{equ:reward}
Reward(e_{ij})=\overline{A} \times \frac{(\alpha-w_{ij})}{\alpha}, if\ w_{ij} \leq \alpha
\end{equation}
\begin{equation}
\label{equ:penalty}
Penalty(e_{ij})=\overline{A} \times \frac{(w_{ij}-\alpha)}{1-\alpha}, if\ w_{ij} > \alpha
\end{equation}
\begin{equation}
\label{equ:MELevel}
a_{ij}=
\begin{cases}
Rounding(\overline{A}-Reward(e_{ij})),  & w_{ij}\leq\alpha\\
Rounding(\overline{A}+Penalty(e_{ij})), & w_{ij}>\alpha
\end{cases}
\end{equation}

Equation \ref{equ:reward} and \ref{equ:penalty} calculate the reward and penalty based on whether $w_{ij}$ is smaller or larger than $\alpha$. $a_{ij}$ is calculated by either subtracting a reward from or adding a penalty to the average distance $\overline{A}$. $\alpha$ controls the degree to which we should reward or penalize an edge. 
Note that other coarsening strategies may also be possible, such as simple linear mapping. Users can dynamically adjust the search granularity using $\alpha$. 



\subsection{Path Scoring}
\label{sec:pathScoring}

In order to avoid comparing massive path scores, we define a different path scoring, which does not add up the edge weights along the path. 

\begin{definition}
\textit{(Path Scoring)}. Given an $l$-length path sequence $p=\{v_1, e_{12} ,v_2,...,v_{l-1},e_{l-1l},v_{l}\}$, its score is recursively defined as follows.
\begin{enumerate}
	\item if $l=1$, $\mathcal{F}(p)=0$,
	\item if $l>1$, $\mathcal{F}(p)=\max\{\mathcal{F}(p\backslash\{e_{l-1l},v_l\})
	,a_{l-1l}\}+1$.
\end{enumerate}
\label{def:pathScoring}
\end{definition} 

According to Definition \ref{def:pathScoring}, the score of a path ending at $v_l$ is calculated by adding 1 to the larger value of $\mathcal{F}(p\backslash\{e_{l-1l},v_l\})$ and $a_{l-1l}$. The former stands for the score of paths ending at $v_{l-1}$, the in-neighbor of $v_l$. The latter is the edge activation level of $e_{l-1l}$. This means we are comparing the path score with the edge activation level. It is reasonable because we scale the edge weights around the average hops $\overline{A}$ of the data graph. More specifically, if an edge $e$ is trivial with a large weight, then it is scaled to a value $a$ larger than $\overline{A}$. Hence, a path passing through $e$ will get a score at least $1 + a$. In other words, the edge activation level can be thought of as the lower bound of costs for passing through that edge. With coarsening process, we can give the same integer score to paths with similar semantic scores. 
Next, we show the relation between the coarse-grained path score (Definition \ref{def:pathScoring}) and the typical fine-grained one. We show that given the edge activation level, we can develop lower and upper bounds for the fine-grained edge weight (Theorem \ref{thm:boundEdgeWeight}). Furthermore, we can obtain an upper bound for the original fine-grained path score given the coarse-grained one (Corollary \ref{thm:boundPathScore}).

\begin{theorem}
Given $a_{ij}$ for edge $e_{ij}$, its original weight $w_{ij}$ can be bounded as follows:
\begin{enumerate}
	\item  $\alpha\frac{a_{ij}-0.5}{\overline{A}}\leq w_{ij} < \alpha\frac{a_{ij}+0.5}{\overline{A}}$,
	when $a_{ij}<Rounding(\overline{A})$,
	\item 
	$\alpha\frac{a_{ij}-0.5}{\overline{A}}\leq w_{ij} < 1+\frac{(a_{ij}+0.5-2\overline{A})(1-\alpha)}{\overline{A}}$, when $a_{ij}=Rounding(\overline{A})$,
	
	\item  $1+\frac{(a_{ij}-0.5-2\overline{A})(1-\alpha)}{\overline{A}} \leq w_{ij} < 1+\frac{(a_{ij}+0.5-2\overline{A})(1-\alpha)}{\overline{A}}$, when $a_{ij}>Rounding(\overline{A})$.
\end{enumerate}

\label{thm:boundEdgeWeight}
\end{theorem}
\begin{proof}
	It is derived due to the corresponding monotonicity and linearity of Equation \ref{equ:reward}, \ref{equ:penalty} and \ref{equ:MELevel}.
\end{proof}

\begin{corollary}
Given an $l$-length path sequence $p=\{v_1, e_{12} ,v_2,\\...,v_{l-1},e_{l-1l},v_{l}\}$ with coarse-grained score  $\mathcal{F}$, we can obtain the upper bound for its fine-grained path score by assuming the edge activation levels are $\{\mathcal{F}-(l-1),...,\mathcal{F}-2,\mathcal{F}-1\}$ w.r.t. the path sequence ordering. 
\label{thm:boundPathScore}
\end{corollary}
\begin{proof}
The upper bound is obtained by setting every edge to its maximum possible activation level. The actual value is calculated by applying Theorem \ref{thm:boundEdgeWeight} on each edge.
\end{proof}

\subsection{CG and RPG Scoring}
\label{sec:graphScoring}
In this section, we introduce the scoring of CGs and RPGs. The final top-$k$ RPGs should reflect close connections in two aspects: (1) among \textit{central keywords} and (2) from \textit{marginal keywords} to \textit{central keywords}.

The score of a CG $\mathcal{G}^c$ is given by the largest path score within it (max. aggregation score) in Equation \ref{equ:cgScore}. Using max. aggregation score has two \textit{benefits}. On one hand, we can \textit{effectively} bound the CG in terms of ``depth''. On the other hand, the search can \textit{efficiently terminated} because the score of a CG monotonically increases w.r.t. its maximum path score.
\begin{equation}
\mathcal{S}^c(\mathcal{G}^c)=\max_{p\in \mathcal{G}^c} \{\mathcal{F}(p)\}
\label{equ:cgScore}
\end{equation}

Similarly, in Equation \ref{equ:mgScore},we use the largest path score from \textit{marginal keywords} to capture the close relationships from \textit{marginal keywords} to \textit{central ones} within an RPG $\mathcal{G}^r = \mathcal{G}^c\cup\mathcal{G}^m$.
\begin{equation}
\mathcal{S}^m(\mathcal{G}^m)=\max_{p\in \mathcal{G}^m} \{\mathcal{F}(p)\}
\label{equ:mgScore}
\end{equation}

Given an RPG $\mathcal{G}^r$, $\mathcal{S}^c(\mathcal{G}^c)$ and $\mathcal{S}^m(\mathcal{G}^m)$ are scores for central and marginal part, respectively. By combining them, we can get a score for an RPG as in Equation \ref{equ:RPGScore}. Note that a multiplicative combination score is also applicable.

\begin{equation}
\mathcal{S}^r(\mathcal{G}^r) =\gamma\mathcal{S}^c(\mathcal{G}^c) + (1-\gamma)\mathcal{S}^m(\mathcal{G}^m)
\label{equ:RPGScore}
\end{equation}

Based on the above ranking schemes, we return top-$k$ RPGs to users and thus solve the Radial Pattern Keyword Search problem (Definition \ref{def:RPKSP}).
In addition, if ties exist when obtaining top-$k$ results, we can break the ties by a re-ranking operation, e.g. using the sum of edge weights. 

\vspace{1mm}
\noindent\textbf{Remarks on avoiding repetitive results.} There may be repetitive results when collecting CGs and RPGs. For CGs, there are two kinds of repetitions. First, a CG completely contains another smaller one. Second, two CGs have the same set of nodes and edges (ignoring directions), but with only different choices of central nodes. To avoid such repetitions, we simply stop expanding a node once it is identified as a central node. 
In addition, RPGs will not be repetitive as long as their corresponding CGs are not.

\subsection{Result Generation}
\label{sec:resultGen}
In this section, we discuss the procedure of result generation based on the above scoring schemes. There exist two issues no matter which combination scoring (i.e. addition or multiplication) we use for Equation \ref{equ:RPGScore}. For effectiveness, it is possible that an RPG is ranked higher with large $\mathcal{S}^c(\mathcal{G}^c)$ and small $\mathcal{S}^m(\mathcal{G}^m)$. As a consequence, such RPGs cannot reflect a close relationship among \textit{central keywords}, on which users focused more. For efficiency, in order to obtain the top-$k$ RPGs, we need to enumerate many possible CGs and check whether they can lead to a better-scored RPG. However, the score in Equation \ref{equ:RPGScore} is not monotonic as we enumerate CGs. Therefore, it can be time-consuming, and we may eventually collect top-ranked results with the aforementioned effectiveness issue. With these considerations, we propose to divide the search into two stages. First, top-$w$ CGs are returned as intermediate results. Then, top-$k$ ($k\leq w$) RPGs are returned as final results. We can collect top-$w$ CGs with highest scores by Equation \ref{equ:cgScore}, which effectively preserves close semantic relationships among \textit{central keywords}. Moreover, we can also avoid enumerating massive \textit{CGs} for finding the optimal RPGs measured by Equation \ref{equ:RPGScore}.

\vspace{1mm}
\noindent\textbf{Remarks}. In fact, the whole process for finding RPGs can be regarded as a two-level beam search. 
Beam search  \cite{Xu:beamSearch,Zhang:maverick} uses BFS (Breadth First Search) to explore the search space. At each level of the BFS tree, only a fixed number of states are preserved for generating the future search states. The leaves of the search tree correspond to the final results.
In our search scheme, CGs correspond to the intermediate search states, i.e. internal nodes. By collecting top-$w$ CGs, we actually preserve $w$ (i.e. the beam width) intermediate states. This is done in the first search level. Then, the final top-$k$ RPGs are generated and ranked by Equation \ref{equ:RPGScore} in the second level.
In practice, we find it not necessary to find hundreds or thousands of CGs with a very large beam width $w$. Such large $w$ may result in CGs whose \textit{central keywords} are loosely connected. Hence, we set $w=k$ so that users can change $w$ and $k$ together while keeping a small $w$ for close connections within \textit{central keywords}.
In the end, the ideal results with both small $\mathcal{S}^c(\mathcal{G}^c)$ (Equation \ref{equ:cgScore}) and $\mathcal{S}^m(\mathcal{G}^m)$ (Equation \ref{equ:mgScore}) can always be captured. 
Moreover, for services like keyword search, users tend to care about only a few top-ranked results. The algorithm presented in Section \ref{sec:implementation} works for the general case $k \leq w$.

\section{Implementation}
\label{sec:implementation}
\subsection{Overview}

\begin{figure}
	\centering
	\includegraphics[width=2.7in]{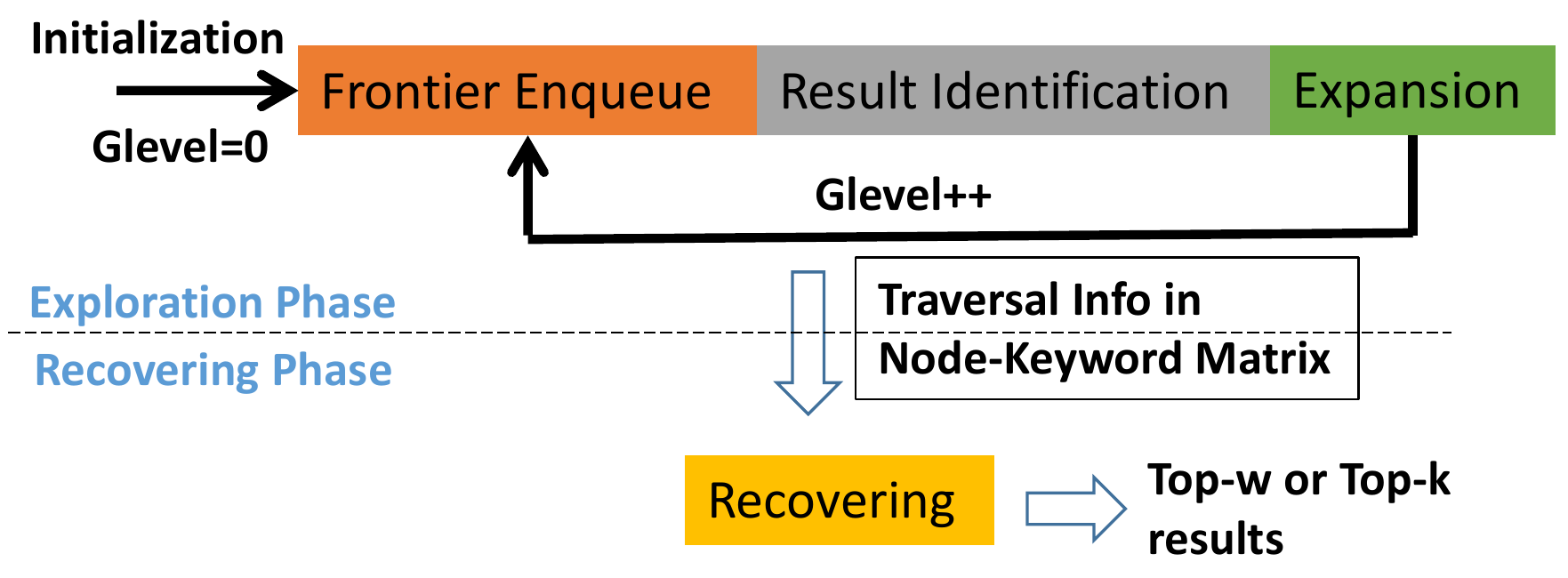}
	\caption{Framework of the two-phase algorithm. 
	Different parallel regions have different colors.
	All threads synchronize between two parallel regions.}
	\label{fig:overview}
\end{figure}

Figure \ref{fig:overview} shows the unified algorithmic framework for finding both the intermediate CGs and the final RPGs.  The respective algorithm runs twice. In the first run, it retrieves a candidate set of CGs as the intermediate results given the input \textit{central keywords}. In the second run, based on the obtained CGs and the \textit{marginal keywords}, the final RPGs are returned.
In this section, we show the implementation details of the algorithm and the respective modifications for retrieving the CGs and RPGs. 

Illustrated in Figure~\ref{fig:overview}, there are two main phases in the algorithm, i.e. \textit{Exploration Phase} and \textit{Recovering Phase}. The Exploration Phase starts exploring the graph from keywords without recording and comparing paths. It contains three main procedures in a loop. The search proceeds in a concurrent BFS-like manner level by level, with one BFS instance for one input keyword. Glevel is the global expansion level to be synchronized after each iteration in the loop. In Exploration Phase, \textit{Node-Keyword Matrix} is used to record the distances from keywords to nodes (Definition \ref{def:distKeyword2Node}), i.e. the smallest path scores (Definition \ref{def:pathScoring}). 
During expansion, this matrix helps resolve reading and writing conflicts in a parallel environment (\textbf{Theorem \ref{thrm:lockfree}}). In Recovering Phase, we make use of this matrix to recover the identified results (\textbf{Theorem \ref{thrm:recover}}) .
The output from Recovering Phase is the set of top-$w$ CGs or top-$k$ RPG w.r.t. the first and the second run of the algorithm.

The two-phase algorithm works in a fork-and-join manner. 
We use OpenMP \cite{Dagum:openmp} to manage the parallelism on CPU since it can reduce thread-creating overhead by keeping a thread pool.  
The parallelism lies in three aspects. First, different keywords expand within their BFS instances in parallel. Second, all frontiers at the same Glevel expand in parallel. Third, neighbors or edges of every frontier node are explored in parallel. 

We store graphs in Compressed Sparse Row (CSR) format. 
\riki runs on the server side and we pre-store the graph structure in the main memory of a CPU or a GPU. Moreover, we use MongoDB \cite{Chodorow:mongodb} to store text information of node entities. The mapping from keywords to nodes is thus handled by MongoDB. In the following, we denote input keyword set as $T$ without distinguishing between \textit{central} or \textit{marginal}.
\subsection{Exploration Phase}

\subsubsection{Initialization}
We initialize and pre-allocate all data structures used in the exploration phase for two reasons. First, with statically allocated memory, we do not need expensive locks for dynamic memory allocation. Second, GPUs only have a limited amount of dynamic memory. Thus, it is necessary to do a statical memory allocation for GPU implementation. 
There are three main structures.

First, we need a \textit{Frontier Flag Array}, denoted as $F$, to indicate all frontiers. These frontier nodes will be expanding at the current Glevel. The search for each keyword can be viewed as an independent BFS instance from source nodes containing that keyword. Thus, we can use one shared Frontier Flag Array \cite{Liu:iBFS} for all keywords. More specifically, if a node is a frontier w.r.t. any keyword (BFS instance), we set its frontier flag to 1. $v_i$ is a frontier only if $F_i=1$, otherwise $F_i=0$. In the beginning, we set $F_i$ of all keyword nodes to 1. In addition, the size of $F$ is $\Theta(|\mathcal{V}|)$.

Second, a \textit{Node-Keyword Matrix}, $\mathcal{H}$, with rows corresponding to nodes and columns keywords. We use $h_{ij}$ (element at i-th row and j-th column in $\mathcal{H}$) to record the distance from keyword $t_j$ to node $v_i$, i.e. $\mathcal{D}(t_j,v_i)$. Note that $\mathcal{D}(t_j,v_i)$ is also the smallest path score from $t_j$ to $v_i$.
The size of $\mathcal{H}$ is $\Theta(|\mathcal{V}||T|)$, where $|T|$ is the number of keywords.

Third, we need a \textit{Joining Flag Array}, denoted by $CF$, to indicate whether a node $v_i$ is reached by all keywords, either central or marginal. If so, $CF_i=1$ otherwise 0. This is done during Result Identification (in gray) as shown in Figure \ref{fig:overview}. The size of $CF$ is also $\Theta(|\mathcal{V}|)$ with one element per node. 

\subsubsection{Frontier Enqueue}
During expansion, the \textit{Frontier Flag Array} $F$ is modified. We extract all frontiers with value 1 in $F$. After frontier enqueue, $F$ is set to all 0 for use in the incoming expansion.
In practice, We find that it is faster to sequentially scan $F$ to extract all frontiers in CPU implementation. Whereas, for GPU, it is faster to scan $F$ and enqueue all frontiers by locked writes with all available warps. One warp consists of usually 32 threads on GPU. The reason is due to bandwidth between main memory (DDR5X) and processors of GPU, which is 480GB/s and 9X faster than our CPU with DDR4. We then set $F$ to all 0 for future use in parallel.

\subsubsection{Result Identification}
\label{sec:resultIdentify}
This stage identifies all top-$w$ CGs or top-$k$ RPGs. The identification is done for all extracted frontiers since only $h_{ij}$'s of frontiers are modified. In searching for \textit{CGs}, we rely on Theorem~\ref{thrm:identifyCG} to do CG identification and score calculation.
The exploration terminates once we collect at least $w$ CGs after a certain Glevel. The correctness based on Equation \ref{equ:cgScore} is guaranteed. This is because increasing Glevel can only result in larger path scores and thus larger scores of CGs.
\begin{theorem}
If $h_{ij}\neq \infty$ for a node $v_i$ and every input \textit{central keyword} $t_j$, then $v_i$ is a central node and locates a CG $\mathcal{G}^c$. In addition, $\mathcal{S}^c(\mathcal{G}^c)=\max_{t_j\in T}{h_{ij}}$.	\label{thrm:identifyCG}
\end{theorem}
\begin{proof}
	Omitted.
\end{proof}

In searching for \textit{RPGs}, the top-$w$ CGs are all recovered. A different $\mathcal{H}$ is initialized for \textit{marginal keywords}. For every CG, we scan the rows (of $\mathcal{H}$) of all its \textit{central keyword} nodes to check if it is reached by all \textit{marginal keywords}. If so, an RPG is identified. 
On CPU, it is easy to do the checking because accessing $\mathcal{H}$ is easy. However, it is much more complicated for GPU implementation because the accessing of those central keyword nodes is not memory aligned. Therefore, we transfer the rows of the central keyword nodes in $\mathcal{H}$ to CPU for the checking process.
The score of an RPG is calculated by Equation \ref{equ:RPGScore}. The early termination for searching top-$k$ RPGs is based on Theorem \ref{thrm:earlyTermination}.
In particular, to guarantee PTC in Definition \ref{def:RPG}, we simply stop expanding a node if it is reached by all \textit{marginal keywords}.

\begin{theorem}
	Suppose at a certain Glevel, we have collected $k$ RPGs and the $k$-th RPG is $\mathcal{G}^r=\mathcal{G}^c\bigcup \mathcal{G}^m$. Let $CG^{set}$ denote the set of all CGs from which no RPG is derived. Then, we can terminate search for top-$k$ RPGs, if the following inequality holds. $$\mathcal{S}^r(\mathcal{G}^r)<=\gamma\min\limits_{\mathcal{G}^*\in CG^{set}}\mathcal{S}^c(\mathcal{G}^*) + (1-\gamma)\mathcal{S}^m(\mathcal{G}^m),$$ where $\gamma$ correspond to that in Equation \ref{equ:RPGScore}.
\label{thrm:earlyTermination}
\end{theorem}
\begin{proof}
	Omitted.
\end{proof}


\begin{algorithm}[t]
	\small
	\SetKwInput{Input}{Input}
	
	\SetKwInput{Output}{Output}
	
	\Input{$\mathcal{H}$, $F$, $CF$, Glevel $l$, parameter $\alpha$, graph adjacency lists}
	
	\Output{Modified $\mathcal{H}$ and $F$}
	
	\tcc{CPU Parallel Level}
	
	\ForEach{frontier $v_f$}
	{
		
		\If{$CF_{v_f} = 1$}{
			
			continue\tcp*{To avoid repetitions.} 
			
		}

		\tcc{GPU parallel level.}
		
		\ForEach{$t_i \in T$}
		{

			$s_f \gets h_{fi}$\tcp*{Path score for $v_f$ from $t_i$.}
			
			\tcc{Check if $v_f$ is a frontier for keyword $t_i$, since $F$ is shared for all keywords.}
			
			\If{$s_f > l$}{
				
				continue\;
				
			}
			
			\ForEach{$v_n \in \mathcal{N}_o(v_f)$}{
				calculate edge activation level $a_{fn}$ w.r.t. $\alpha$\;

				\If{$a_{fn} > l$}{
					
					$F_f \gets 1$\tcp*{$v_f$ continues being frontier.}
					
					continue\;	
					
				}

				$s_n \gets h_{ni}$\tcp*{Path score for $v_n$ from $t_i$.}
				
				\If{$s_n \neq \infty$}{
					
					continue\;
					
				}

				\tcc{Ready to expand to $v_n$.}
				
				$h_{ni}\gets l + 1$\;
				
				$F_n \gets 1$\;

			}	
			
		}
		
	}

	\KwRet\; 
	
	\caption{Expansion Procedure}
	
	\label{alg:expansion}
	
\end{algorithm}
\subsubsection{Expansion}
\label{sec:expansion}
We first present the expansion behavior of frontiers in Definition \ref{def:expansionBehavior}. Based on this behavior, we develop the relationship between Glevel $l$ and the distances from keywords to nodes, as in Theorem \ref{thrm:distGlevel}. Eventually, we can find the distances from keywords to nodes. These distances are stored in $\mathcal{H}$ and used to score CGs and RPGs in the Result Identification step.
\begin{definition}
(\textit{Expansion Behavior}).
Given a frontier $v_f$ and a keyword $t_i$, the \textit{expansion behavior} at Glevel $l$ is as follows.
We scan the out-edges from $v_f$, and for every neighbor $v_n$, we skip $v_n$ if it is already reached by $t_i$. Otherwise, we expand $v_f$ to $v_n$ only if $a_{fn}\leq l$.
\label{def:expansionBehavior}
\end{definition}
\begin{theorem}
Suppose $v_f$ expand to $v_n$ at Glevel $l$ w.r.t. keyword $t_i$, then $\mathcal{D}(t_i,v_n) = l+1$ with the expansion behavior in Definition \ref{def:expansionBehavior}. Furthermore, $h_{ni}=l+1$ in Node-keyword Matrix.
\label{thrm:distGlevel}
\end{theorem}
\begin{proof}
The expansion behavior leads to the path scores in Definition \ref{def:pathScoring}.
The details are omitted due to limited space.
\end{proof}

As shown in Algorithm \ref{alg:expansion},  the expansion procedure first takes a frontier $v_f$. Then, it handles the expansion of $v_f$ with respect to every keyword $t_i$ from line 4. Line 2 works for different purposes for CG search and RPG search. For the former, line 2 is used to avoid repetitions in the way we stop expanding from identified central nodes. For the latter, the goal is to guarantee PTC in Definition \ref{def:RPG}. From line 8 to 17, we handle all neighbors of $v_f$. Theorem \ref{thrm:lockfree} shows the lock-free property.

\begin{theorem}
	In Algorithm \ref{alg:expansion}, the writes and reads are all lock-free without affecting final results.
	\label{thrm:lockfree}
\end{theorem}

\begin{proof}
	For writes, all writes are writing the same value, thus they are lock-free, including line 11, 16 and 17 in Algorithm \ref{alg:expansion}. For reads, we show that all if-conditions return the correct truth-value from unlocked reads, including line 5 and 13. At line 6, we do not expand $v_f$ if it is not a frontier at the current level for keyword $t_i$. $s_f$ may be set to $l+1$ in current level by other threads or remain $\infty$. At line 14, as long as $s_n\neq \infty$, we do not expand to it since $v_n$ may be reached in current level or previous ones.
\end{proof}

\textbf{Scheduling}. There are a few differences in the parallel strategies of CPU and GPU implementations. CPU implementation is more coarse-grained (line 1).  We use one thread for different frontiers for all keywords. In comparison, GPU parallelism starts at line 4 in Algorithm \ref{alg:expansion}. We use one warp to handle one keyword expansion of one frontier node. At line 8, the threads within the warp handle all neighbors of the corresponding frontier in parallel w.r.t. keyword $t_i$, so that these threads can share memory access to the adjacency lists. There are two main reasons for the different parallel strategies. First, the threads of a GPU is far more than multi-core CPUs, thus GPU implementations can be more fine-grained. Second, execution divergence can hurt efficiency very much on GPUs. Within a warp of a GPU, the execution of threads may diverge when there are if-branches. This divergence can hurt the performance and higher level of parallelism results in more execution divergence. In comparison, CPU threads are more powerful in handling such divergence. Hence, we use different parallel granularities.

\textbf{Load balancing}. On GPU, the load balancing problem mainly comes from different number of neighbors of nodes. We handle all neighbors of a node by one warp. However, our edge weighting strategy gives massive same-labeled edges with large weights. Hence, processing of these edges quickly skips the loop at line 10. Moreover, nodes with massive same-labeled edges will have fewer chances to participate in search. Therefore, it will not cause severe load balancing problem. On CPU, we use the \textit{dynamic} scheduling feature of OpenMP \cite{Dagum:openmp}. OpenMP can automatically handle the load balancing by letting idle threads take more jobs.

\textbf{Time and space complexity}. For time complexity, the expansion process of one keyword is quite like BFS expansion. For standard BFS, a node is only visited once. Thus, the respective time complexity is $\mathcal{O}(|\mathcal{V}|+|\mathcal{E}|)$, where $|\mathcal{V}|$ and $|\mathcal{E}|$ denote the number of nodes and edges. However, it is different for Algorithm \ref{alg:expansion}. The reason is that nodes and edges may be visited more than once if the edge activation level is not matched by the current Glevel. This adds up to the time complexity. As a result, the time complexity of Algorithm \ref{alg:expansion} in sequential execution is $\mathcal{O}(|\mathcal{V}||\mathcal{E}||T|l_{max})$, where $|T|$ is the number of keywords and $l_{max}$ denotes the maximum level in consideration. Furthermore, since the parallelism has no overlap, the complexity for Algorithm \ref{alg:expansion} in parallel is $\mathcal{O}(\frac{|\mathcal{V}||\mathcal{E}||T|l_{max}}{Tnum})$, where $Tnum$ is the number of threads. For space complexity, the major cost arises from graph storage $\Theta(|\mathcal{V}|+|\mathcal{E}|)$, the Frontier Flag Array $\Theta(|\mathcal{V}|)$, the Joining Flag Array $\Theta(|\mathcal{V}|)$, and the node-keyword matrix $\Theta(|T||\mathcal{V}|)$. Altogether, the space complexity is $\mathcal{O}(|T||\mathcal{V}|+|\mathcal{E}|)$.

\subsection{Recovering Phase}
\label{sec:recoveringPhase}
\begin{algorithm}[t]
	\small
	\SetKwInput{Input}{Input}
	\SetKwInput{Output}{Output}	
	\Input{$\mathcal{H}$, identified central nodes, graph adjacency lists, keyword set $T$}
	
	\Output{top-$w$ CGs}
	
	
	Initialize a container $\textit{R}$ to keep result CGs\;
	
	\ForEach{$v_c$ in identified central nodes}{
		\tcc{Create one queue for each keyword $t_i\in T$.}
		
		Add $v_c$ to $|T|$ queues $q_1, q_2,... ,q_{|T|}$\;
		
		\While{$\exists q_i \neq \emptyset$}{
			\ForEach{$t_i \in T$}{
				$v_q\gets q_i.next()$\;
				
				\ForEach{$v_n \in \mathcal{N}_i(v_q)$}{
					\If{$h_{qi}=1+\max\{a_{nq}, h_{ni}\}$}{
						
						Recover node $v_n$ and edge $e_{nq}$\; 	
						
						Append $v_n$ into $q_i$ if $h_{n_i}\neq 0$\;
						
					}
					
				}		
				
			}	
			
		}
		
		Append the extracted graph to $R$\tcp*{with lock}	
		
	}	
	\KwRet\;	
	\caption{Recovering Procedure of Central Graphs}	
	\label{alg:recover}		
\end{algorithm}

The Recovering Phase is a reverse procedure of Exploration Phase. Specifically, we recover all paths covered by CGs and RPGs. 
In fact, to recover those paths is equivalent to recover all their edges. In other words, we should be able to recover $v_n$ from $v_q$, if $v_n$ expanded to $v_q$ via edge $e_{nq}$ during expansion.
We show that this is achieved with Lemma \ref{lem:recover}.
\begin{lemma}
	$h_{qi}=\max\{h_{ni},a_{nq}\} + 1$ if and only if $v_n$ expanded to $v_q$ w.r.t. keyword $t_i$ during exploration phase.
	\label{lem:recover}
\end{lemma} 
\begin{proof}
	It is derived from Definition \ref{def:pathScoring}. 
\end{proof}

We show the recovering procedure of CGs in Algorithm \ref{alg:recover}. The recovering procedure of RPGs is similar. Only two changes need to be made. First, we change $v_c$ (line 2) to be the set of central keyword nodes of a CG. Second, for every marginal keyword, we only recover the paths with smallest scores ending at the central keyword nodes. These paths correspond to $SP(m_i,\mathcal{V}_\mathcal{C})$ for all marginal keywords in Definition \ref{def:RPG}.  
The recovering procedure is implemented only on CPUs, because it requires dynamic memory allocation to record all result subgraphs. Moreover, the recovering process has much execution divergence since it needs to decide whether or not to extract a certain node. As shown in Algorithm \ref{alg:recover}, at line 2 we loop through all CGs from Section \ref{sec:resultIdentify}. From line 4 to 10, it is a standard independent BFS for each keyword. 

\textbf{Load balancing}. We parallelize Algorithm \ref{alg:recover} in a coarse-grained way from line 2. If we choose to parallelize the inner loops, we will need to synchronize at the end of every iteration of the outer loop. It is very expansive. In addition, we use \textit{dynamic} scheduling feature provided by OpenMP \cite{Dagum:openmp}.

\textbf{Time and space complexity}. Since recovering CGs and RPGs has the same complexity, we focus on the complexity of Algorithm \ref{alg:recover} for CGs. For time complexity, we can view the recovering process of one keyword as a standard BFS instance for a CG. Therefore, the complexity in sequential execution is $\mathcal{O}(|R||T|(|\mathcal{V}|+|\mathcal{E}|))$, where $|T|$ and $|R|$ are the number of keywords and identified CGs, respectively. As a result, the parallel time cost is $\mathcal{O}(\frac{|R||T|(|\mathcal{V}|+|\mathcal{E}|)}{Tnum})$ due to non-overlap execution, where $Tnum$ is the number of threads. For space complexity, since we continue using data structures of Algorithm \ref{alg:expansion}, we need $\mathcal{O}(|T||\mathcal{V}|+|\mathcal{E}|)$ space cost. Moreover, the space for container $R$ in Algorithm \ref{alg:recover} is bounded by $\mathcal{O}(|R||\mathcal{V}|)$, where $|R|$ is the number of CGs to recover. Therefore, the total space cost is $\mathcal{O}((|R|+|T|)|\mathcal{V}|+|\mathcal{E}|)$.

\section{Experiments}
\label{sec:experiment}

\subsection{Settings}
\subsubsection{Competitors}
\begin{itemize}[leftmargin=3mm,noitemsep]
	\item \textbf{\riki.} It is the proposed approach. In particular, we denote the corresponding CPU and GPU versions as \textbf{\riki-CPU} and \textbf{\riki-GPU}, respectively. 
	
	\item \textbf{\wiki} \cite{Yang:wikisearch}. It is a parallel keyword search approach that can run in real time. We implement both CPU and GPU parallel version of \wiki. They are denoted as \textbf{\wiki-CPU} and \textbf{\wiki-GPU}, respectively.
	
	\item \textbf{BANKS-\rom{2}} \cite{Kacholia:bidirection}. It is an established and representative work of typical keyword search approach, which is also a revised version of BANKS-\rom{1} \cite{Aditya:banks}. As for the reason for choosing BANKS-\rom{2}, please refer to the discussions in Section \ref{sec:relatedWork}.
\end{itemize}

  
\subsubsection{Datasets and Queries}
\begin{table}
	\centering
	\small
	\caption{Wikidata Dumps}
	\begin{tabular}{c|c|c||c|c}
		dataset & \# nodes & \# edges & $\overline{A}$& Deviation\\ \hline
		WikiSmall & 15.1M & 124M & 3.87 & 0.81 \\
		WikiLarge & 30.6M & 271M & 3.68 &0.98\\ 
	\end{tabular}
	\label{tble:datasets}	
\end{table}
As shown in Table \ref{tble:datasets}, we use two publicly available Wikidata dumps, WikiSmall and WikiLarge, from \cite{Yang:wikisearch}. The last two columns show the sampled average shortest distances and the sample deviation. We sample ten thousand pairs of nodes for estimation of $\overline{A}$.

We obtain keyword queries from the keyword lists given by papers accepted in AAAI 2013 and 2014. They are available at UCI machine learning repository \cite{Dua:uci}. These keywords naturally serve as reasonable queries. One one hand, they directly describe the topics of the scientific articles. On the other hand, they are commonly given in an ordering with decreased importance and relevancy.
Thus, we can use the first several keywords as \textit{central keywords} and the rest as \textit{marginal keywords}. More details of query generation are shown at the beginning of Section \ref{sec:efficiency} and \ref{sec:effectiveness}.
 
\subsubsection{Parameters to Study and Experiment Platform.} 
We list the parameters to study with default values in Table \ref{tble:param}. 
\begin{table}
	\centering
	\small
	\caption{Parameters to study}
	\begin{tabular}{ccc}
		parameter & meaning & default\\\hline
		Topk     & the number of RPGs to be returned & 20 \\
		$\alpha$ & the parameter introduced in Section \ref{sec:edgeActivation} & 0.5 \\
		cknum & the number of \textit{central keywords} & 2\\
		mknum & the number of \textit{marginal keywords} & 4\\
		Knum  & cknum + mknum & 6 \\
		Tnum  & the number of CPU threads & 30\\
		  
	\end{tabular}

\label{tble:param}
\end{table} 
 
All programs are written in C++ 4.8.5 with Cuda 8.0 and OpenMP 3.1. We turn on -O3 flag for compilation. All tests were run on a single machine with CentOS 7.0 and 52-core Intel(R) Xeon(R) Platinum 8170 CPUs @ 2.1GHz. The main memory of CPU is of 1 TB DDR4 with 64-bit data width. The GPU is GTX 1080 Ti with 11 GB memory with DDR5X which has 352-bit memory bus width.  
In all experiments, we set the time limit as 500 seconds. If such limiting parameters are reached, we return whatever is obtained. We set maximum allowed Glevel to 20.

\subsection{Efficiency Studies}
\label{sec:efficiency}
\begin{figure*}[h]
\centering
	\subfloat[WikiSmall]{\includegraphics[width=3.35in]{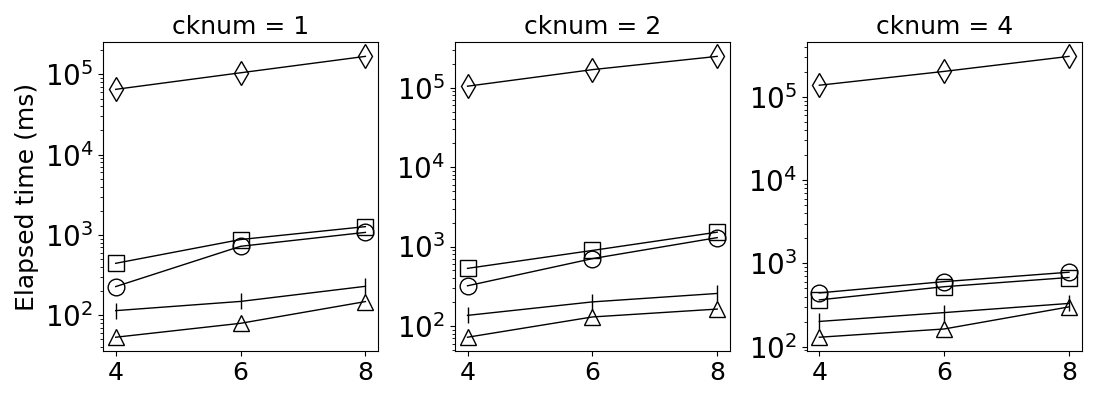}\label{subfig:wiki2017vknum}}
	\subfloat[WikiLarge]{\includegraphics[width=3.35in]{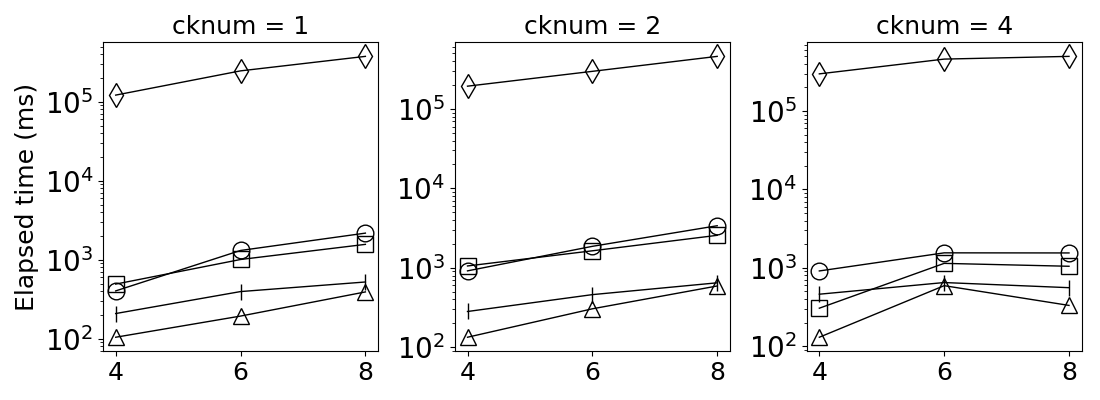}\label{subfig:wiki2018vknum}}

\caption{Vary Knum on WikiSmall and WikiLarge. X-axis shows mknum. Each sub-figure has a fixed cknum. \textbf{Due to limited space, we use the same legend as Figure \ref{fig:varyTopk}}.}
\label{fig:varyKnum}
\end{figure*}

\begin{figure*}[h]
	\centering
	\includegraphics[width=6.6in]{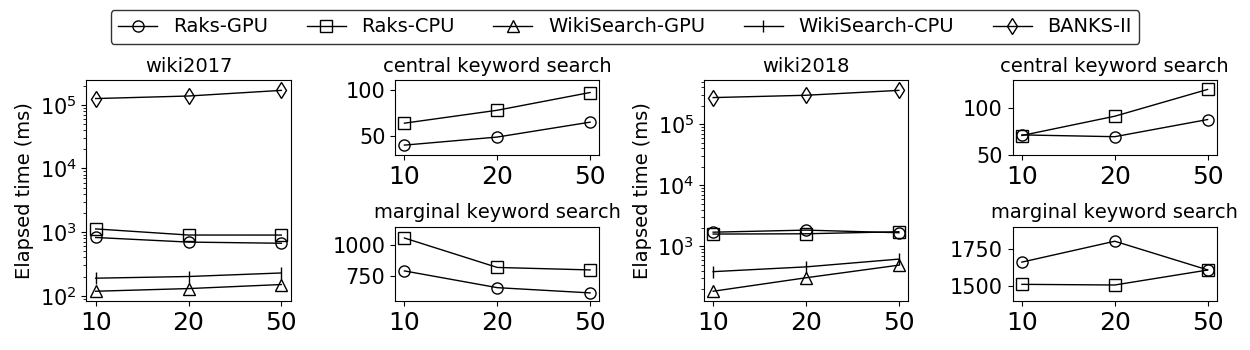}

\caption{Vary Topk on WikiSmall and WikiLarge. X-axis shows Topk values. The two sub-figures on the right of each figure show the broken-down time of central keyword search (i.e. the first run) and marginal keyword search (i.e. the second run) for respective datasets.}
\label{fig:varyTopk}
\end{figure*}

\begin{figure}[h]
	\centering
	\subfloat[Vary Tnum. X-axis shows Tnum.]{\includegraphics[width=\linewidth]{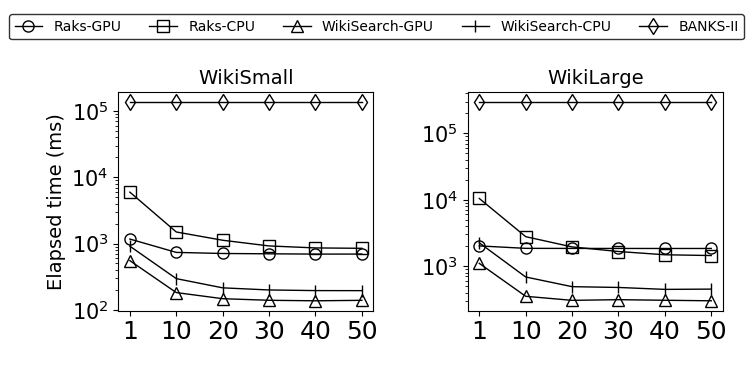}\label{subfig:varyTnum}}\\
	\subfloat[Vary $\alpha$. X-axis shows values of $\alpha$]{\includegraphics[width=\linewidth]{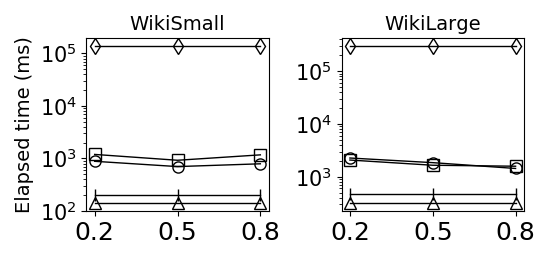}\label{subfig:varyAlpha}}
	
\caption{Vary Tnum and $\alpha$ on WikiSmall and WikiLarge.}
\label{fig:TnumAndAlpha}
\end{figure}

\begin{figure}[h]
	\centering
	\includegraphics[width=2.3in]{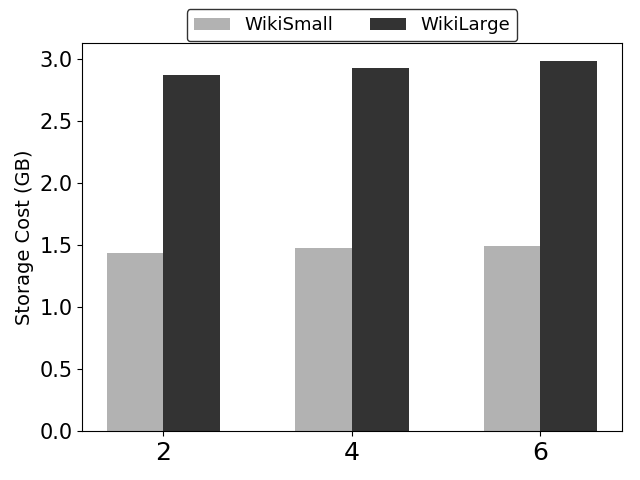}
	\caption{Peak storage cost of \riki of two datasets on GPU. We fix cknum=4. X-axis shows mknum (from 2 to 6).}
	\label{fig:storageCost}
\end{figure}

\textbf{Query Generation}. 
To study the effect of the number of keywords, we take the combination of three settings of cknum (1, 2, 4) and mknum (2, 4, 6). We randomly select 50 queries for every different setting of Knum (3, 4, 5, 6, 7, 8, 10), which is the sum of cknum and mknum. Then we use the first several keywords as \textit{central keywords} and the rest as \textit{marginal keywords}. For example, to get cknum=1 and mknum=2, we take the selected 50 queries of Knum=3. The running time is averaged over the 50 queries.

\vspace{1mm}\noindent
\textbf{Exp-1 (vary Knum)}. In Figure \ref{fig:varyKnum}, we show the effect of Knum on different datasets. Since our competitors do not implement radial-pattern search, we compare the efficiency w.r.t. the same Knum, i.e. total number of keywords. From the result, we can conclude that \riki is at least 2 orders faster than \textit{BANKS-\rom{2}} which is a CPU sequential algorithm. \textit{BANKS-\rom{2}} fails producing real-time response for several reasons. Compared to \textit{BANKS-\rom{1}}, \textit{BANKS-\rom{2}} adds forward search to avoid wasting time on solely backward probing. Whereas, when the graph size gets larger, the forward search also needs to probe massive nodes. It hurts the performance a lot. In comparison, \riki runs much faster, as many edges are explored in parallel.

We also notice that \wiki always gives the best performance. The reason is that \riki coarsens the graph in terms of edges. Nevertheless, \wiki coarsens the graph in terms of nodes. It is based on a node-wighted graph, thus it is even more coarse-grained. Specifically, take \textit{human} node as an example. In node-weighted graph, this node with millions of same-labeled edges (\textit{instance of}) is given a high weight. Thus, it is hardly reachable during the search. In contrast, we model the graph as edge weighted. Nodes like \textit{human} may be reached by informative edges that do not have massive same-labels (e.g. \textit{part of} and \textit{humanity}). Once \textit{human} node is activated, it tries to expand to its neighbors. As the number of its neighbors is huge, it tests all out-going edges. This process may hurt the performance. Although \riki spends more time compared to \wiki, it can still finish search within around one second.

As Knum increases, the processing time of \riki only increases by hundreds of milliseconds. Whereas, the running time \textit{BANKS-\rom{2}} increases tens of seconds. We also found that mknum has a larger influence on efficiency than cknum. The reason is that if there are more \textit{central keywords}, we have more connecting points for \textit{marginal keyword} nodes to connect. This makes it easier for searching RPGs. On the other hand, if there are more \textit{marginal keywords}, we have to connect more \textit{marginal keyword} nodes to a CG. This inevitably incurs more costs.

\vspace{1mm}\noindent
\textbf{Exp-2 (vary Topk)}. In Figure \ref{fig:varyTopk}, we show the effects of varying Topk. Note that we set the beam width $w$ of CG candidates to be the same as top-k value of RPGs. As a result, users can use top-k to directly modify the number of CGs to consider. It is clear that when Topk increases, \textit{BANKS-\rom{2}} and \wiki has an increasing processing time. In contrast, \riki has a trend of decreasing processing time. To explain this, we also show the broken-down time for the whole processing. The broken-down time includes \textit{central keywords} search and \textit{marginal keywords} search time. The two parts  of time clearly show that when Topk increases, the search time of CG increases. However, the search time of \textit{marginal keywords} has a decreasing trend. The reason is that when Topk increases, we collect more candidate CGs. This results in more joining points (i.e. candidate CGs) for \textit{marginal keywords}. In other words, it becomes easier for marginal keywords to be connected to a CG. In addition, the major time is spent on marginal keyword search.
Hence, the total processing time has a decreasing trend when Topk increases and the marginal keyword search time decreases.
It is expected that when search time of CGs continue increasing, the total time will eventually increase.

\vspace{1mm}\noindent
\textbf{Exp-3 (vary Tnum)}. In Figure \ref{subfig:varyTnum}, we vary the number of CPU threads from 1 to 50. Note that when Tnum=1, \riki-CPU becomes a sequential processing. For \riki-GPU, increasing of CPU threads only benefits the recovering phase, because other steps are on GPU. This is why the increasing of Tnum only benefits \riki-GPU to a small extent. In comparison, the processing time of \riki-CPU can decrease by one order of magnitude from Tnum=1 to Tnum=50. It can even beat \riki-GPU when Tnum=50 on WikiLarge. Reasons are twofolds. On one hand, CPU traversal is faster with more threads. On the other hand, as the graph size becomes larger compared to WikiSmall, there are more divergent execution within a warp on GPU, because more edges are examined to know whether they can be expanded w.r.t. a keyword. As a result, \riki-GPU slows down with larger graph size and \riki-CPU speeds up with more threads.

\vspace{1mm}\noindent
\textbf{Exp-4 (vary $\alpha$)}. In Figure \ref{subfig:varyAlpha}, we can see that the running time is relatively stable while varying $\alpha$, but there exists a decreasing trend of processing time when $\alpha$ gets larger. We found that with small $\alpha$ a large portion of edges are inaccessible at the early stage of search. This causes the processing to unnecessarily halt just to wait for a larger Glevel. In contrast, with larger $\alpha$, more edges are accessible early, leading to faster results retrieving. 

From all experiments, we can see that \riki-GPU tends to be faster than \riki-CPU. We owe this to the larger bandwidth between main memory and processors (DDR5X) of GPU, which is over 350 GB/s versus 55 GB/s on CPU with DDR4.

\vspace{1mm}\noindent
\textbf{Exp-5 (Running Storage Cost)}. Note that all text information of nodes and edges can be kept in external storages such as MongoDB \cite{Chodorow:mongodb}. Since CPU has relatively large memory, we focus on studying the running storage cost on GPU of \riki. In Figure \ref{fig:storageCost}, we set cknum=4 and vary mknum from 2 to 6, so the storage cost is maximum in terms of \textit{central keywords}. We only record the peak running storage cost on GPU. The storage cost on GPU mainly comes from the adjacency lists of the data graph and the node-keyword matrix. From Figure \ref{fig:storageCost}, we can see that increasing of several keywords just causes a very small increase of memory usage. This is because adding one keyword is equivalent to adding one column to node-keyword matrix. Note that there are around 11 GB memory of our GPU, indicating that we can handle graphs with an even larger size. 

In addition, the \textbf{memory transfer cost} between CPU (host) and GPU (device) mainly comes from copying the node-keyword matrix for both \textit{central keywords} and \textit{marginal keywords}. The copy is done only once from device to host. The time cost is in total just around 10 ms on WikiSmall and 20 ms on WikiLarge. Thus, it is not a main concern for \riki.

\subsection{Effectiveness Studies}
\label{sec:effectiveness}
\begin{table*}
	\centering
	\small
	\caption{Keyword Queries. kwfSmall and kwfLarge denote the overall average keyword frequencies.}
	\begin{tabular}{c||c|c||c|c} 
		Query ID & \textit{central keywords} & \textit{marginal keywords} & kwfSmall & kwfLarge\\ \hline
		1    & vascular& network blood flow  & 6059 & 87102\\ \hline
		2    &  machine learning&  multi-Label learning learning algorithm  &51 &931\\ \hline
		3    &  tensor & multilinear algebra   kernel method & 78 &1267\\\hline
		4    &  text mining &  grounding   computer vision & 56 & 242\\\hline
		5    & cognition affect &  decision making cooperation  &1226 &20939 \\\hline
		6    &  causal inference causal Reasoning causality & selection bias case-control studies  & 43 & 568\\\hline
		7    &  social media &  temporal dynamics topic modeling user group  & 94 &700\\\hline
		8    &  Gaussian Processes regression&  Bayesian summarization model description  & 198 & 3859\\\hline
		9    & structure learning Bayesian learning &  graphical models deep belief networks  & 5 & 58\\\hline
		10    &  disease mapping & spatial scan mobile data  & 7 & 51\\\hline
		11    &  machine learning  & online learning endowment effect & 52 & 877\\\hline
		12    &planning & Bayesian optimization Gaussian Processes &439 & 5757 \\\hline
		13    & temporal dynamics & social influence social networks  & 56 & 807\\\hline
		14    &  feature selection &  dimensionality reduction semi-supervised learning  & 17 & 359\\\hline
		15    &  Sparse coding transfer learning &  supervised learning classification  &559&6199\\
		 &&  Support Vector Machine  &  & \\
	\end{tabular}
	\label{tble:queries}
\end{table*}
\begin{figure*}[ht]
	\centering
	\subfloat[WikiSmall]{\includegraphics[width=\linewidth]{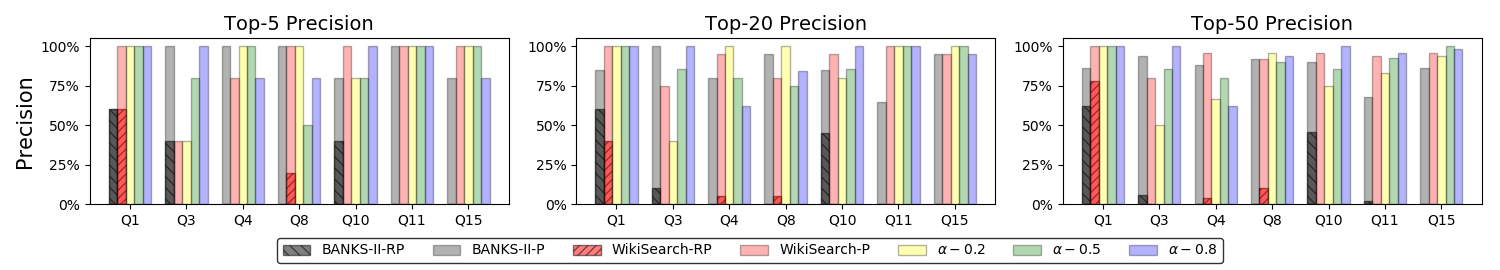}\label{subfig:wiki2017topkPrecisionPart1}}\\
	\subfloat[WikiLarge]{\includegraphics[width=\linewidth]{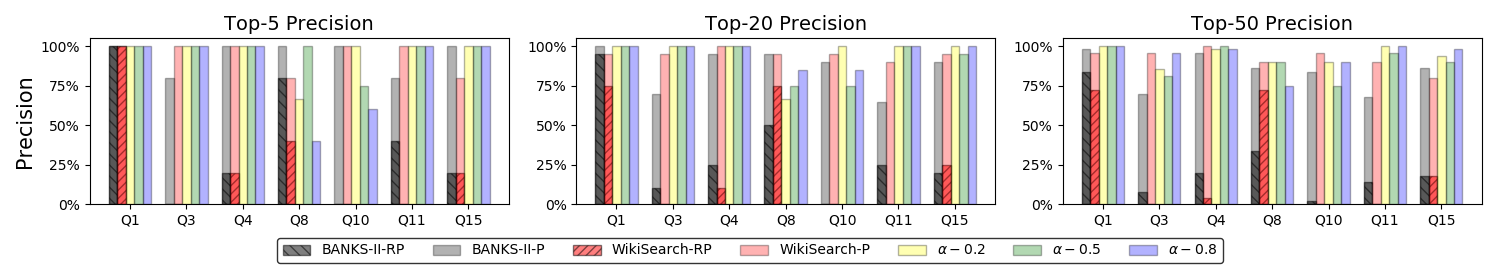}\label{subfig:wiki2018topkPrecisionPart1}}
\caption{Top-k Precision and Ratio of Radial Pattern Results. The results include \textit{BANKS-\rom{2}-RP} (ratio of radial-pattern results by \textit{BANKS-\rom{2}}), \textit{BANKS-\rom{2}-P} (top-k precision), \textit{WikiSearch-RP} (ratio of radial-pattern results), \textit{WikiSearch-P} (top-k precision) and three settings of $\alpha$ (0.2, 0.5, 0.8) of \riki. For \textit{BANKS-\rom{2}} and \wiki, we show their ratio of radial-pattern results inside the bar of the respective top-k precision. In particular, we do not show the ratio of radial pattern results of \riki, which is always equal to the top-k precision.}
	
\label{fig:topkPrecisionPart1}
\end{figure*}
\textbf{Query generation}. In this section, we evaluate two aspects in terms of effectiveness of \riki. First, the RPGs returned by \riki are relevant. Second, \riki can return RPGs with radial pattern shapes while the competitors can only find such results by chance. 
We select 15 representative queries for effectiveness tests and manually separate them into \textit{central} and \textit{marginal keywords}, as shown in Table \ref{tble:queries}. Note that testing effectiveness by selected keyword queries is a commonly used way in the field of keyword search \cite{Aditya:banks,He:blinks,Kacholia:bidirection,Kargar:RCliques,Li:ease,Li:groupsteiner,Yang:FindingPatternsKeyword}.
We use the following two metrics in the evaluation.
\begin{itemize}[leftmargin=3mm,noitemsep]
	\item \textbf{Top-k precision.} It is used in many experiment studies for keyword search \cite{Aditya:banks,He:blinks,Kacholia:bidirection,Li:ease,Kargar:RCliques}. Top-k precision measures the ratio of relevant results within top-k retrieved results. The relevancy of a result is judged by researchers in our group. We follow three rules for judgment. First, the results should refer to related concepts or topics. Second, repetitive results are counted only once. These results differ at only one or two nodes and produce little new information. Third, results with very long paths and massive nodes are regarded as irrelevant results. Nodes within these results are not closely related mainly due to short-cut effects.
	\item \textbf{Ratio of Radial Pattern Results.}
	It is the ratio of the top-$k$ results which satisfy PTC in Definition \ref{def:RPG}. By showing this ratio, we validate our claim that the typical keyword search approaches can mistakenly put user focused keywords far apart in the result. And less focused keywords constitute the central part.
	This motivates us to develop \riki, which can always return radial pattern graphs as results. 
\end{itemize}

For clearance, we show one part of results in Figure \ref{fig:topkPrecisionPart1}. The other part is in Figure \ref{fig:topkPrecisionPart2}, which are relatively better due to higher keyword frequencies.

From an overall perspective, we summarize three observations. First, larger frequency of input keywords leads to better answers, due to more counterparts of keywords in data graphs. This makes it easier to form an result graph. Q1 is a typical example, and more examples are in Figure \ref{fig:topkPrecisionPart2}. Second, all top-5 results are shown to be good, which validates that our search scheme can return effective answers for a small top-$k$ ($w=k$) setting. Last but not the least, \riki can effectively return RPGs in comparison with \textit{BANKS-\rom{2}} and \wiki. The two competitors even failed returning any, such as Q3, Q4, Q8, Q11 and Q15 in Figure \ref{fig:topkPrecisionPart1} for some settings. This validates our claim that typical keyword search methods may miss connections among focused keywords.

From a detailed perspective, we find that $\alpha$ can affect the top-k precision of results returned by \riki. For example, in Figure \protect\ref{subfig:wiki2017topkPrecisionPart1}, the top-k precision of \riki for Q3 drops when $\alpha$ is relatively small. The reason is that the two marginal keywords of Q3 only appear 3 and 5 times in WikiSmall dataset. As a result, the number of effective paths (without passing shortcuts) from these \textit{marginal keywords} nodes is limited. Smaller $\alpha$ postpones traversing through shortcut edges. As a result, most detailed effective paths from \textit{marginal} keywords are covered before shortcuts. On one hand, these paths meet before hitting the \textit{central keywords}, which do not lead to RPGs. On the other hand, due to their limited number, these paths can be very long before reaching the \textit{central keywords}. Consequently, we find that the top-2 results of Q3 by \riki with small $\alpha$ tend to include most of the relevant paths. Thus, the rest of the results tend to be irrelevant due to their long path connections.
In addition to Q3, Q10 and Q11 have a similar problem as shown in Figure \protect\ref{subfig:wiki2017topkPrecisionPart1}. 

Another observation is that the top-k precision of some queries drops when $\alpha$ becomes larger, such as Q4 and Q8 in Figure \protect\ref{subfig:wiki2017topkPrecisionPart1} and Q8 in Figure \protect\ref{subfig:wiki2018topkPrecisionPart1}. The reason is that when using a larger $\alpha$, edges tend to be activated earlier in search. These edges may include some meaningless shortcuts, which lead to less informative results such as a \textit{conference} node connecting two papers. Due to shortcut effects, the result graphs may contain massive nodes, thus labeled as irrelevant.

In sum, we can see that typical methods like \textit{BANKS-\rom{2}} and \wiki can possibly miss the connections among central keywords. In contrast, \riki can always return results without losing focus.
In addition, \riki remains competitive in terms of top-k precision. Although $\alpha$ may affect the effectiveness of results, users can always get an immediate response by re-submitting the query with a different $\alpha$ due to the fast processing of \riki.

\begin{figure*}[ht]
	\centering
	\subfloat[WikiSmall]{\includegraphics[width=\linewidth]{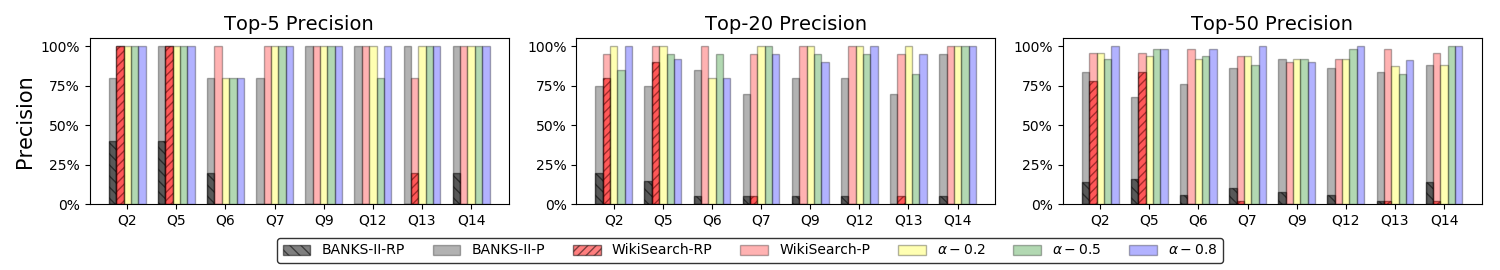}\label{subfig:wiki2017topkPrecisionPart2}}\\
	\subfloat[WikiLarge]{\includegraphics[width=\linewidth]{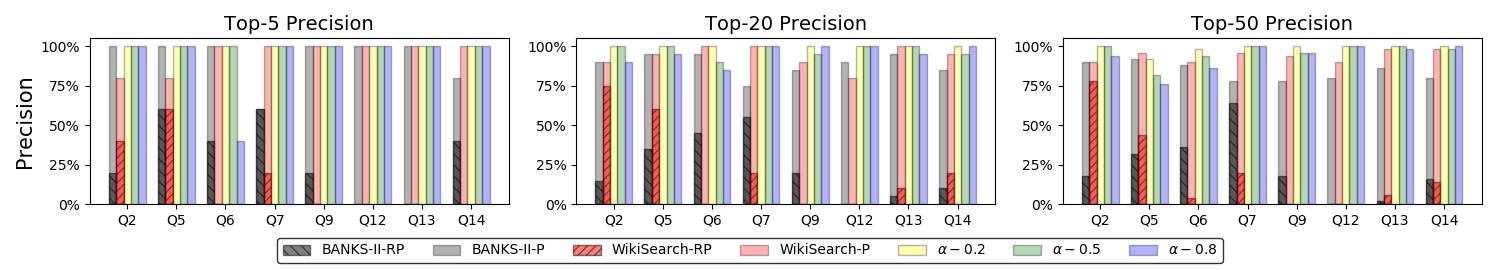}\label{subfig:wiki2018topkPrecisionPart2}}
	\caption{Top-k Precision and Ratio of Radial-pattern Results. Complementary results of Figure \ref{fig:topkPrecisionPart1}.}
	\label{fig:topkPrecisionPart2}
\end{figure*}
\section{Related Work}
\label{sec:relatedWork}
Our work is mainly related to keyword search area. Keyword search is popular for a long time, since early works like DBXplorer \cite{Agrawal:dbxplorer} and DISCOVER \cite{Hristidis:discover}. Following this line of research, there are many other works, like \cite{Qin:TenThousandSQLs,Qin:powerOfRDBMS,Qin:queryingCommunities,Yu:kwsearchDB,Markowetz:KwSearchDataStream,Luo:Spark}. These works focus on relational databases and the corresponding schema graph which is usually small. In particular, \cite{Qin:TenThousandSQLs} introduces a parallel approach, which depends mainly on relational database operations.

There are also many approaches \cite{Aditya:banks,Dalvi:keywordSearchExternal,Ding:min-cost-topk,He:blinks,Kacholia:bidirection,Le:ScalableKeywordSearch,Li:groupsteiner,Yang:wikisearch} that can work on a graph structured data. Among these, BANKS-\rom{1} \cite{Aditya:banks} and BANKS-\rom{2} \cite{Kacholia:bidirection} propose a backward search algorithm. These works all try to find a tree-shaped answer that closely relates and covers all keywords. \cite{Dalvi:keywordSearchExternal,He:blinks,Le:ScalableKeywordSearch} provide disk solutions for huge graphs. And, BLINKS \cite{He:blinks} builds index to accelerate search. \cite{Ding:min-cost-topk,Li:groupsteiner} use dynamic programming to find optimum steiner tree. However, \cite{Ding:min-cost-topk} generates top-k results progressively and is rather slow pointed by \cite{Li:groupsteiner}. Moreover, it is not clear how to obtain top-k by the latter. EASE \cite{Li:ease} and RCliques \cite{Kargar:RCliques} provides graph-shaped results. However, as data size grows rapidly, these index-based approaches suffer from scalibility issues. \wiki \cite{Yang:wikisearch} is a recent work that focuses on accelerating keyword search by modern hardwares, e.g. GPUs and multi-core CPUs. 

However, none of the above works distinguish between \textit{focused} and \textit{less focused} keywords.
In addition, there are works allowing users to input preferences for keyword search \cite{Duan:KeywordSearchProduct,Stefanidis:Perk}. However, \cite{Duan:KeywordSearchProduct} studies the problem on returning a product entity with user preferences. \cite{Stefanidis:Perk} propose to let users input preferences directly when submitting the keyword queries and the results are ranked according to the preferences. Furthermore, \cite{Stefanidis:Perk} only targets at relational databases.
Differently, \riki considers focused and less focused keywords.
There is no preference in \textit{marginal keywords}. Moreover, the relationships between \textit{marginal} and \textit{central keywords} are explicitly discovered as connection paths.

\section{Conclusion}
\label{sec:conclusion}
In this paper, we propose an efficient parallel keyword search engine, called \riki. It takes as input two sets of keywords, namely \textit{central keywords} and \textit{marginal keywords}. The corresponding results can capture close relationships among \textit{central keywords}, while revealing how \textit{marginal keywords} are related to or connected via \textit{central keywords}. In order to improve the efficiency, we devise novel weighting and scoring schemes along with a unified algorithmic framework. 
\bibliographystyle{ACM-Reference-Format}
\bibliography{Raks}

%
% ...

\end{document}